\documentclass[runningheads]{llncs}
\usepackage{amssymb}
\usepackage{enumerate}
\usepackage{proof}
\usepackage{stmaryrd}
\usepackage{times}
\usepackage{amssymb}
\usepackage{graphicx}
\usepackage{textcomp}
\usepackage{amsmath}
\usepackage{stmaryrd}
\usepackage{url}
\usepackage{setspace}
\urldef{\mailsa}\path|mmh.thu@gmail.com|
\urldef{\mailsb}\path|pennyshaq@gmail.com|

\newcommand{\mf}{\mathfrak}

\newcommand{\mc}{\mathcal}

\newcommand{\imp}{\rightarrow}

\newcommand{\Imp}{\Rightarrow}

\newcommand{\D}{\Diamond}
\newcommand{\B}{\Box}

\newtheorem{thm}{Theorem}

\newtheorem{lem}[thm]{Lemma}

\title{The Computational Compexity of Decision Problem in Additive Extensions of Nonassociative Lambek Calculus}
\author{Zhe Lin\inst{1} \and Minghui Ma\inst{2}}
\institute{Institute of Logic and Cognition Sun Yat-sen University, Guangzhou, China\\
Faculty of Mathematics and Computer Science, Adam Mickiewicz University, Pozna\'{n}, Poland.\\
\email{pennyshaq@gmail.com}
\and
Institute for Logic and Intelligence, Southwest University,\\
Beibei District, Chongqing, 400715, China.\\
\email{mmh.thu@gmail.com}}

\begin{document}
\maketitle
\begin{abstract}
We analyze the complexity of decision problems for Boolean Nonassociative Lambek Calculus admitting empty antecedent of sequents ($\mathsf{BFNL^*}$), and the consequence relation of Distributive Full Nonassociative Lambek Calculus ($\mathsf{DFNL}$). We construct a polynomial reduction from modal logic K into $\mathsf{BFNL^*}$. As a consequence, we prove that the decision problem for $\mathsf{BFNL^*}$ is PSPACE-hard. We also prove that the same result holds for the consequence relation of $\mathsf{DFNL}$, by reducing $\mathsf{BFNL^*}$ in polynomial time to $\mathsf{DFNL}$ enriched with finite set of assumptions. Finally, we prove analogous results for variants of $\mathsf{BFNL^*}$, including $\mathsf{BFNL_e^*}$ ($\mathsf{BFNL^*}$ with exchange), modal extensions of $\mathsf{BFNL_i^*}$ and $\mathsf{BFNL_{ei}^*}$ for $i\in\{\mathsf{K,T,K4,S4,S5}\}$.
\end{abstract}

\section{Introduction and Preliminaries}
Nonassociative Lambek Calculus ($\mathsf{NL}$) was introduced by Lambek \cite{jm2} as a variant of Lambek Calculus $\mathsf{L}$ \cite{JM}. Many variants of $\mathsf{L}$ and $\mathsf{NL}$ were studied in the last decades. $\mathsf{L}$ extended with conjunction ($\wedge$) and disjunction $(\vee$) was introduced in \cite{Ka}. $\mathsf{NL}$ with $\wedge$, $\vee$ satisfying the distribution law ($\mathsf{DFNL}$), and $\mathsf{DFNL}$ with a boolean negation $\neg$ ($\mathsf{BFNL}$), were studied in \cite{bus3,bus4}, where it was proved that the consequence relations of both systems are decidable, and that the categorial grammars based on them generate context-free languages. The proof of decidability is based on the proof of the finite embeddability property in \cite{bus4}. The decidability of the latter one was later shown again in terms of relational semantics in \cite{DSk,KF13}. There are also many complexity results for $\mathsf{L}$, $\mathsf{NL}$ and their variants \cite{dest,Penu2,YS,HTK}. The most outstanding one is that $\mathsf{L}$ is NP-complete \cite{Penu2}.

In this paper we analyze the complexity of the decision problem of $\mathsf{BFNL^*}$ ($\mathsf{BFNL}$ admitting empty antecedent of sequents), and that of the consequence relation of $\mathsf{DFNL}$. The main result is that the decision problems for both $\mathsf{BFNL^*}$ and the consequence relation of $\mathsf{DFNL}$ are PSPACE-hard. Both results were claimed first in \cite{bus7} and the latter one was proved by Buszkowski using a different method in an unpublished paper. The relational semantics for $\mathsf{BFNL^*}$ in \cite{KF13} is essentially used in our proof. We take some techniques and  notations from  \cite{Ver1,KF13}. We also study the consequence relations for logics. Put it differently, we consider logics enriched with (finitely many) assumptions which are simple sequents but not closed under uniform substitutions. Hereafter, we denoted logic $\mathbf{L}$ enriched with set of assumptions $\mathrm{\Phi}$ by $\mathbf{L} \mathrm{(\Phi)}$.

This paper is organized as follows. In what follows of this section, we introduce some notations and remind the sequent system of $\mathsf{BFNL^*}$ and the complexity results for normal modal logics. In section 2, we construct a polynomial reduction from modal logic $\mathsf{K}$ into $\mathsf{BFNL^*}$, which yields the PSPACE-hardness of the decision problem for $\mathsf{BFNL^*}$. In section 3, we show the decision problem for $\mathsf{DFNL}(\mathrm{\Phi})$ is PSPACE-hard by reducing $\mathsf{BFNL^*}$ first to $\mathsf{BDFNL}(\mathrm{\Phi})$ (Bounded Distributive Full Nonassociative Lambek Calculus), and then to $\mathsf{DFNL}(\mathrm{\Phi})$  in polynomial time. In section 4, we extend our results to some variants of $\mathsf{BFNL^*}$, including $\mathsf{BFNL^*}$ enriched with exchange,modalities, constant $1$ and any combination of them.

Now let us fix our notations. The language $\mc{L}_\mathsf{K}(\mathsf{Prop})$ of modal logic consists of a set $\mathsf{Prop}$ of propositional letters, connectives $\bot, \wedge, \vee,\supset$ and an uary modal operator $\D$.
The set of all modal formulae is defined by the following inductive rule:
\[
A::= p\mid \bot\mid A\wedge B \mid A\vee B\mid A\supset B\mid \D A,~~p\in\mathsf{Prop}
\]
Define $\neg A:= A\supset \bot$, $\B A:= \neg \D\neg A$ and $A\equiv B := (A\supset B)\wedge (B\supset A)$.

Let $\mf{M}=(W,R,V)$ be a Kripke model, where $W$ is a nonempty set of states, $R$ is a binary relation over $W$, and $V:\mathsf{Prop}\imp \wp(W)$ (powerset of $W$) is a valuation function. The notion of {\em truth} of a modal formula $\mf{M},w\models A$ is defined recursively as follows:
\begin{enumerate}
\item[] $\mf{M},w\models p$ iff $w\in V(p)$.
\item[] $\mf{M},w\not\models \bot$
\item[] $\mf{M},w\models A\vee B$ iff $\mf{M}, w\models A$ or $\mf{M}, w\models B$.
\item[] $\mf{M}, w \models A\wedge B$ iff $\mf{M},w\models A$ and $\mf{M}, w\models B$.
\item[] $\mf{M},w\models A\supset B$ iff $\mf{M}, w\not \models A$ or $\mf{M},w\models B$.
\item[] $\mf{M}, w\models \D A$, if there exists $u\in W$ such that $Rwu$ and $\mf{M},w\models A$.
\end{enumerate}
A modal formula $A$ is {\em valid}, if it is true at every state in all models.

The minimal normal modal logic $\mathsf{K}$ is axiomatized by the following axiom schemata and inference rules (\cite{BDV01}):
\begin{enumerate}
\item[{$\bullet$}]All instances of propositional tautologies.
\item[{$\bullet$}]$\Box(A\supset B)\supset (\Box A\supset \Box B)$
\item[{$\bullet$}](MP) from $\vdash A \supset B$ and $\vdash A$ infer $\vdash B$.
\item[{$\bullet$}](Nec) from $\vdash A$ infer $\vdash\Box A$.
\end{enumerate}
The modal logic $\mathsf{K}$ is sound and complete, i.e., a modal formula $A$ is provable in $\mathsf{K}$ iff $A$ is valid.

The PSPACE-hardness of the validity problem of modal logic $\mathsf{K}$ was settled first by Ladner \cite{Lad77}. Let us recall this thereom from \cite{BDV01} (Theorem 6.50).

\begin{thm}[Lander's Theorem]\label{theorem:lander}
If $\mathcal{S}$ is a normal modal logic such that $\mathsf{K}\subseteq \mathcal{S} \subseteq \mathsf{S4}$ then $\mathcal{S}$ has a PSPACE-hard satisfiability problem. Moreover, $\mathcal{S}$ has PSPACE-hard validity problem.
\end{thm}

Now we recall some basic notions and sequent system for $\mathsf{BFNL^*}$.
Let $\mc{L}_{\mathsf{BFNL^*}}(\mathsf{Prop})$ be the language of $\mathsf{BFNL^*}$ built from the set $\mathsf{Prop}$ of propositional letters by Lambek connectives $/, \backslash,\cdot$, and propositional connectives $\wedge,\vee,\bot,\top$ and $\neg$.
The set of all $\mc{L}_{\mathsf{BFNL^*}}(\mathsf{Prop})$-formulae is defined by the following inductive rule:
\[
A::= p\mid \bot\mid A\wedge B \mid A\vee B\mid A\backslash B\mid A/B\mid A\cdot B,~~p\in\mathsf{Prop}.
\]
The set of all formula trees is defined by the rule
\[
\Gamma::= A\mid \Gamma\circ \Delta
\]
where $A$ is a $\mc{L}_{\mathsf{BFNL^*}}(\mathsf{Prop})$-formula. Each formula tree $\Gamma$ is associated with a formula $\varphi(\Gamma)$ defined recursively as follows: $\varphi(A) = A$; $\varphi(\Gamma\circ\Delta) = \varphi(\Gamma)\cdot\varphi(\Delta)$.

Sequents are of the form $\Gamma \Rightarrow A$ where $\Gamma$ is a formula tree and $A$ is a formula. By $\Phi\vdash_{S} \Gamma \Rightarrow A$ we mean sequent is derivable from $\Phi$ in system $S$. The sequent calculus $\mathsf{BFNL^*}$ consists the following axioms and rules:
 \begin{displaymath}
   (\mathrm{Id}) \quad A \Rightarrow A \quad \mathrm{(D)}\quad A\wedge (B\vee C) \Rightarrow (A\wedge B) \vee (A\wedge C).
   \end{displaymath}
\begin{displaymath}
    (\bot) \quad \Gamma[\bot]\Rightarrow A \quad(\top)\quad \Gamma \Rightarrow \top
    \end{displaymath}
   \begin{displaymath}
   \mathrm{(\neg 1)} \quad A\wedge \neg A \Rightarrow \bot\quad \mathrm{(\neg 2)} \quad\top \Rightarrow A\vee \neg A.
   \end{displaymath}
   \begin{displaymath}
   (\mathrm{\backslash L}) \quad\frac{\Delta \Rightarrow A  \quad \Gamma[B] \Rightarrow C}{\Gamma[\Delta \circ (A \backslash B)] \Rightarrow C} \quad(\mathrm{\backslash R})\quad \frac{A \circ \Gamma \Rightarrow B}{\Gamma \Rightarrow A \backslash B}
   \end{displaymath}
   \begin{displaymath}
   (\mathrm{/L})\quad \frac{\Gamma[A] \Rightarrow C\quad \Delta \Rightarrow B}{\Gamma[(A/B) \circ \Delta] \Rightarrow C} \quad (\mathrm{/R})\quad \frac{\Gamma \circ B \Rightarrow A}{\Gamma \Rightarrow A/B}
   \end{displaymath}
   \begin{displaymath}
   (\mathrm{\cdot L})\quad\frac{\Gamma[A \circ B] \Rightarrow C}{\Gamma[A \cdot B] \Rightarrow C} \quad (\mathrm{\cdot R}) \quad\frac{\Gamma \Rightarrow A \quad \Delta \Rightarrow B}{\Gamma \circ \Delta \Rightarrow A \cdot B}\quad
   (\mathrm{Cut}) \quad\frac{\Delta \Rightarrow A \quad \Gamma[A] \Rightarrow B}{\Gamma[\Delta] \Rightarrow B}
   \end{displaymath}
 \begin{displaymath}
   \mathrm{(\wedge L)}\quad\frac{\Gamma[A_i]\Rightarrow B}{\Gamma[A_1\wedge A_2]\Rightarrow B}{~(i=1,2)}\quad\mathrm{(\wedge R)}\quad\frac{\Gamma\Rightarrow A\quad \Gamma \Rightarrow B}{\Gamma \Rightarrow A\wedge B}
   \end{displaymath}
   \begin{displaymath}
   (\mathrm{\vee L})\quad\frac{\Gamma[A_1]\Rightarrow B\quad \Gamma[A_2]\Rightarrow B}{\Gamma[A_1\vee A_2] \Rightarrow B}\quad(\mathrm{\vee R)}\quad \frac{\Gamma \Rightarrow A_i}{\Gamma\Rightarrow A_1\vee A_2}{~(i=1,2)}
   \end{displaymath}
The $\Gamma$ in ($\backslash\mathrm{R}$) and ($\mathrm{/R}$) can be empty. Notice that the following facts hold in $\mathsf{BFNL^*}$:
\begin{itemize}
\item[(1)] $\vdash_{\mathsf{BFNL^*}} \neg \bot\Leftrightarrow \top $ and  $\vdash_{\mathsf{BFNL^*}} \neg \top \Leftrightarrow \bot$.
\item[(2)] $\vdash_{\mathsf{BFNL^*}} A\Leftrightarrow\neg\neg A$.
\item[(3)] $\vdash_{\mathsf{BFNL^*}}  \neg(A\wedge B)\Leftrightarrow \neg A\vee\neg B$ and $\vdash_{\mathsf{BFNL^*}} \neg(A\vee B)\Leftrightarrow \neg A\wedge\neg B$.
\item[(4)] $\vdash_{\mathsf{BFNL^*}} A\wedge (B\vee C) \Leftrightarrow (A\wedge B) \vee (A\wedge C)$ and $\vdash_{\mathsf{BFNL^*}} A\vee (B\wedge C) \Leftrightarrow (A\vee B) \wedge (A\vee C)$.
\item[(5)] $\vdash_{\mathsf{BFNL^*}} m\cdot(A\vee B)\Leftrightarrow (m\cdot A) \vee (m\cdot B)$.
\item[(6)] if $\vdash_{\mathsf{BFNL^*}} A\Imp B$, then $\vdash_{\mathsf{BFNL^*}} \neg B \Imp \neg A$.
\item[(7)] if $\vdash_{\mathsf{BFNL^*}} A\Rightarrow B$ then $\vdash_{\mathsf{BFNL^*}} \Rightarrow \neg A\vee B$
\item[(8)] if $\vdash_{\mathsf{BFNL^*}}  A\Leftrightarrow B$, then $\vdash_{\mathsf{BFNL^*}}  C\Leftrightarrow C'$ where $C'$ is obtained from $C$ by replacing one or more occurrences of $A$ by $B$ in $C$.
\end{itemize}
It is easy to prove (1),(2), (3),(4), (5),(6), and(8). Here we only show (7). Assume $A\Rightarrow B$. By ($\mathrm{\vee R}$), one gets $A\Rightarrow B\vee \neg A$. Since $\neg A\Imp B\vee \neg A$ is provable in $\mathsf{BFNL^*}$, by ($\mathrm{\vee L}$), one obtains $A\vee \neg A\Imp B\vee \neg A$. Then since $\Rightarrow \top$, $\top \Rightarrow A \vee \neg A$ are instances of axioms, by (Cut), one gets $\Imp \neg A\vee B$.

Moreover, $\mathsf{BFNL^*}$ admits the extended subformula property, i.e., if a sequent $\Gamma\Imp A$ is provable in $\mathsf{BFNL^*}$, then there exists a derivation of $\Gamma\Imp A$ such that all formulae appearing in the derivation belong to the set of all subformulae in $\Gamma\Imp A$ and closed under $\wedge$, $\vee$ and $\neg$.

There is also relational semantics for $\mathsf{BFNL^*}$ (\cite{KF13}). A $\mathsf{BFNL^*}$-model is a ternary relational model $\mf{J}=(W,R,V)$ where $W$ is a non-empty set of states, $R$ is a ternary relation over $W$, and $V$ is a valuation from $\mathsf{Prop}$ to the power set of $W$. The satisfiability relation $\mf{J}, u\models A$ between a relational model with a state and a $\mathsf{BFNL^*}$-formula is defined recursively as follows:
\begin{enumerate}
\item[] $\mf{J},u\models p$ iff $u\in V(p)$.
\item[] $\mf{J},u\not\models \bot$ and $\mf{J},u\models\top$.
\item[] $\mf{J}, u\models A\cdot B$, if there are $v, w\in W$ such that $R(u,v, w)$, $\mf{J}, v\models A$ and $\mf{J},w\models B$.
\item[] $\mf{J}, u\models A/B$, if for all $v, w\in W$ such that $R(w,u, v)$, $\mf{J}, v\models B$ implies $\mf{J},w\models A$
\item[] $\mf{J}, u\models A\backslash B$, if for all $v, w\in W$ such that $R(v,w, u)$, $\mf{J}, w\models A$ implies $\mf{J},v\models B$.
\item[] $\mf{J},u\models A\wedge B$ iff $\mf{J},u\models A$ and $\mf{J},u\models B$,
\item[] $\mf{J},u\models A\vee B$ iff $\mf{J},u\models A$ or $\mf{J},u\models B$.
\item[] $\mf{J},u\models \neg A$ iff $\mf{J},u\not\models A$.
\end{enumerate}
The notions of satisfiability, validity and semantic consequence relation are defined as usual (\cite{KF13}). By $\models_{\mathsf{BFNL^*}} A$ we mean that $A$ is valid in all $\mathsf{BFNL^*}$-models. For any sequent $\Gamma\Imp A$,  we say that $\Gamma\Imp A$ is true at a state $u$ in the model $\mf{J}$ (notation: $\mf{J},u\models\Gamma\Imp A$), if $\mf{J},u\models\varphi(\Gamma)$ implies $\mf{J},u\models A$. A sequent $\Gamma\Imp A$ is true in $\mf{J}$ (notation: $\mf{J}\models\Gamma\Imp A$), if $\mf{J},u\models\Gamma\Imp A$ for all states $u$ in $\mf{J}$.

The Hilbert style system for $\mathsf{BFNL^*}$ is equivalent to $\mathsf{PNL}$ in \cite{KF13}. From the results in \cite{KF13} that $\mathsf{BFNL^*}$ is sound and complete under the relational semantics. The following soundness theorem can be easily verified by induction on the length of derivation.
 \begin{thm}\label{thm:comBFNL^*}
For any $\mc{L}_{\mathsf{BFNL^*}}(\mathsf{Prop})$-formula $A$, if $\vdash_{\mathsf{BFNL^*}} \Rightarrow A$, then $\models_{\mathsf{BFNL^*}} A$.
 \end{thm}

\section{PSPACE-hard Decision Problem in $\mathsf{BFNL}^*$}
In this section, we reduce the validity problem of modal logic $\mathsf{K}$, which is PSPACE-complete, to the validity problem of $\mathsf{BFNL^*}$ so that we prove the PSPACE-hardness of the latter problem. Thus the PSPACE-hardness of the decision problem in $\mathsf{BFNL}^*$ follows.
Now let us consider the embedding of modal logic $\mathsf{K}$ into $\mathsf{BFNL^*}$. Let $P\subseteq Prop$ and $m\not\in P$ for a distinguished propositional letter. Define a function $(.)^\dag$: $\mc{L}_\mathsf{K}(\mathsf{P})\imp \mc{L}_{\mathsf{BFNL^*}}(\mathsf{P}\cup\{m\})$ recursively as follows:
\begin{displaymath}
p^\dag=p\quad \bot^\dag=\bot\quad
(A\wedge B)^\dag= A^\dag\wedge B^\dag \quad (A\vee B)^\dag= A^\dag\vee B^\dag
\end{displaymath}
\begin{displaymath}
(A\supset B)^\dag= \neg A^\dag \vee B^\dag \quad
(\neg A)^\dag= \neg A^\dag\quad
(\Diamond A)= m\cdot A^\dag
\end{displaymath}

Let $\mf{M}=(W,R,V)$ be a binary Kripke model with a valuation $V:\mathsf{Prop}\imp \wp(W)$. We define a $\mathsf{BFNL^*}$-model $\mf{J}^\mf{M}=(W',R',V')$  from $\mf{M}$ as follows:
\begin{enumerate}
\item[](1)~~$W'=\{w_1,w_2\mid w\in W\}$
\item[](2)~~$R'=\{\langle w_1, w_2, u_1\rangle | \langle w,u\rangle \in R\}$
\item[](3)~~$V'(p) = \{w_1,w_2\mid w\in V(p)\}$ for $p\in\mathsf{Prop}$; and $V'(m) = W'$.
\end{enumerate}
Intuitively, for each state $w$ in the binary model we make two copies $w_1$ and $w_2$, and then define the tenary relation among copies according to the original binary relation $R$. Note that the order of $w_1$ and $w_2$ makes sense in the ternary relation.

\begin{lem}\label{lemma:truth}
Let $\mf{M}=(W,R,V)$ be a binary Kripke model and $\mf{J}^\mf{M}=(W',R',V')$.
For any $w\in W$ and modal formula $A$, $\mf{M},w\models A$ iff $\mf{J}^\mf{M}, w_1\models A^\dag$.
\end{lem}
\begin{proof}
By induction on the complexity of modal formula $A$.
The atomic and boolean cases are easy by the construction of $\mf{J}^\mf{M}$ and the inductive hypothesis. For $A=\D B$, assume
$\mf{M},w\models \D B$. Then there exists $u\in W$ such that $Rwu$ and $\mf{M},u\models B$. Since $Rwu$, we get $R'(w_1,w_2,u_1)$.
By inductive hypothesis, $\mf{J}^\mf{M},u_1\models B^\dagger$. Hence $\mf{J}^\mf{M},w_1\models m\cdot B^\dagger$.
Conversely, assume $\mf{J}^\mf{M},w_1\models m\cdot B^\dagger$. Then there exists $u_1\in W'$ such that $R'(w_1,w_2,u_1)$, $\mf{J}^\mf{M}, w_2\models m$ and $\mf{J}^\mf{M},u_1\models B^\dag$. By inductive hypothesis, $\mf{M},u\models B$. By the construction of $\mf{J}^\mf{M}$, we get $Rwu$. Hence $\mf{M},w\models \Diamond B$. \qed
\end{proof}

\begin{lem}\label{lemma:com}
For any modal formula $A$, if $\vdash_{\mathsf{BFNL^*}} \Imp A^\dag$, then $\vdash_{\mathsf{K}} A$.
\end{lem}
\begin{proof}
Assume $\not\vdash_{\mathsf{K}} A$. Then there is a binary Kripke model $\mf{M}$ such that $\mf{M}\not\models A$. By lemma \ref{lemma:truth}, $\mf{J}^\mf{M} \not\models A^\dag$. Hence, by theorem \ref{thm:comBFNL^*}, we get $\not\vdash_{\mathsf{BFNL^*}} \Imp A^\dag$. \qed
\end{proof}

\begin{lem}\label{lemma:KBFNL}
For any modal formula $A$, if $\vdash_\mathsf{K} A$, then $\vdash_{\mathsf{BFNL^*}} \Rightarrow A^\dag$.
\end{lem}
\begin{proof}
We proceed by induction on the length of proof in $\mathsf{K}$. It suffices to show all axioms and inference rules of $\mathsf{K}$ are admissible in $\mathsf{BFNL^*}$ w.r.t the translation $\dag$. Obviously the translations of all instances of propositional tautologies are provable in $\mathsf{BFNL^*}$.
Consider $(\Box(A\supset B)\supset (\Box A\supset \Box B))^\dag=m\cdot(A\wedge \neg B)\vee (m\cdot(\neg A))\vee(\neg (m\cdot(\neg B))$. Since $A\wedge B\Rightarrow B$, by Fact (6), one gets $\neg B \Rightarrow (\neg A \vee \neg B)$. Hence by monotonicity of $\cdot$, one gets $m\cdot (\neg B)\Rightarrow (m\cdot(\neg A\vee \neg B))$. Then by Fact (7), one gets $\Rightarrow \neg(m\cdot (\neg B))\vee (m\cdot(\neg A\vee \neg B))$. Since $(A\vee \neg A)\Leftrightarrow \top$ are instances of axioms, by Fact (8), one gets $(m\cdot(\neg A\vee \neg B))\Leftrightarrow m\cdot((A\vee \neg A)\wedge(\neg A\vee \neg B))$.
By Fact (4) and (8), one gets $m\cdot((A\vee \neg A)\wedge(\neg A\vee \neg B))\Leftrightarrow m\cdot((A\wedge\neg B)\vee \neg A)$. Again, by Fact (5), one can prove $ m\cdot((A\wedge\neg B)\vee \neg A)\Leftrightarrow (m\cdot(A\wedge \neg B))\vee ( (m\cdot (\neg A)))$. Hence one gets $\Rightarrow m\cdot(A\wedge \neg B)\vee  (m\cdot (\neg A))\vee (\neg (m\cdot (\neg B)))$.

Let us consider the rule (MP). Assume $\vdash_{\mathsf{BFNL^*}} \Rightarrow A^\dag$ and $\vdash_{\mathsf{BFNL^*}}\Rightarrow  (A\supset B)^\dag$, which is equal to $\vdash_{\mathsf{BFNL^*}}\Rightarrow \neg (A^\dag)\vee B^\dag$. We need to show $\vdash_{\mathsf{BFNL^*}} \Rightarrow  B^\dag$. By ($\neg 1$), ($\bot$) and ($\mathrm{Cut}$), one gets $ A^\dag\wedge \neg (A^\dag) \Rightarrow  B^\dag$. By $\mathrm{(\wedge L)}$,  one gets $A^\dag\wedge B^\dag \Rightarrow  B^\dag$. Then, by ($\mathrm{\vee L}$), one gets $(A^\dag\wedge \neg (A^\dag))\vee (A^\dag\wedge B^\dag) \Rightarrow B^\dag$. Then by ($\mathrm{D}$) and ($\mathrm{Cut}$), one gets $ A^\dag\wedge(\neg (A^\dag)\vee B^\dag) \Rightarrow B^\dag$. Clearly, by assumptions and ($\mathrm{\wedge R}$), one gets $\Rightarrow  A^\dag\wedge(\neg A^\dag\vee B^\dag)$, which yields $\Rightarrow B^\dag$ by ($\mathrm{Cut}$).

Finally consider the rule (Nec). Assume $\vdash_{\mathsf{BFNL^*}} \Rightarrow A^\dag$. We need to show $\vdash_{\mathsf{BFNL^*}} \Rightarrow \neg (m\cdot(\neg A^\dag))$. Then by ($\top$) and Fact (1) and (6), one gets $\neg (A^\dag)\Rightarrow \bot$. By $(\cdot\mathrm{R})$, one gets $m\cdot\neg (A^\dag) \Rightarrow m\cdot \bot$. Then by ($\bot$), ($\mathrm{\cdot L}$) and ($\mathrm{Cut}$), one gets $m\cdot \neg (A^\dag) \Rightarrow \bot$. Hence by ($\neg R$), one gets $\Rightarrow \neg(m\cdot(\neg A^\dag))$. \qed
\end{proof}

Lemma \ref{lemma:com} and Lemma \ref{lemma:KBFNL} lead to the following theorem.
\begin{thm}
$\vdash_\mathsf{K} A$ iff $\vdash_{\mathsf{BFNL^*}} A^\dag$
\end{thm}
Obviously the reduction is in polynomial-time. Now by Lardner's theorem (\ref{theorem:lander}), one gets the following theorem.

\begin{thm}
The validity problem of  $\mathsf{BFNL}^*$ is PSPACE-hard.
\end{thm}
\begin{thm}\label{theorem:comBFNL}
The decision problem in $\mathsf{BFNL}^*$ is PSPACE-hard.
\end{thm}

\begin{remark}
The embedding function $(.)^\flat$ in \cite{Ver1} is also defined to translate the behaviour of $\Diamond$ in term of $\cdot$, which is used in \cite{ja2} to prove the context-freeness of $\mathsf{L(\Diamond)}$ ($\mathsf{L}$ enriched with an unary modal operation and its residual $\Box^\downarrow$). The embedding function $(.)^\flat$ differs from our $(.)^\dag$ in the following two clauses: $(\Diamond A)^\flat= m\cdot A^\flat\cdot n$ and $(\Box^\downarrow A)^\flat =m\backslash A^\flat/n$. It requires two arguments $m, n$ to translate the behaviour of $\Diamond$ since the modal fomulae under consideration contain $\backslash$ or $/$ and the systems admit associativity, while both cases do not occur in our setting.
\end{remark}

\section{PSPACE-hard Decision Problem in $\mathsf{DFNL}(\mathrm{\Phi})$}
In this section, we prove that $\mathsf{DFNL}(\mathrm{\Phi})$  has PSPACE-hard decision problem. In what follows, we assume that $\Phi$ is a finite set of simple sequents, i.e., sequents of the form $A\Imp B$ where $A,B$ are formulae. $T$ denotes a set of formulae. By a $T$-sequent we mean a sequent such that all formulae occurring in it belong to $T$. We write $\Phi\vdash_{S} \Gamma \Rightarrow_T A$, if $\Gamma \Rightarrow A$ has a deduction from $\Phi$ in the system $S$ consisting of $T$-sequents only. 

Our first step of reduction is a polynominal one from $\mathsf{BFNL}^*$ to $\mathsf{BDFNL}^*\mathrm{(\Phi)}$ (i.e., bounded distributive full nonassociative Lambek calculus enriched with assumptions). Let us introduce some notions first.

Let $T$ be a set of formulae containing $\top$ and $\bot$ and closed under taking subformulae.  By $c(T)$ we mean the closure of $T$ under $\vee$ and $\wedge$. It is obvious that $c(T)$ is closed under taking subformulae. We define $T^\sim=T\cup\{p_B| B\in T\}$. Furthermore, we define the function $(.)^\sim$ : $c(T)\hookrightarrow c(T^\sim)$ inductively as follows:
\begin{itemize}
\item[(1)] $\top^\sim=\bot$ and $\bot^\sim=\top$;
\item[(2)] $A^\sim=p_A$ for $A\in T$ and $A\not= \top, \bot$;
\item[(3)] $ (A\wedge B)^\sim =A^\sim \vee B^\sim$ and $(A\vee B)^\sim =A^\sim\wedge B^\sim$.
\end{itemize}
Define $\Psi[T] = \{A\wedge p_A\Rightarrow \bot \mid A\in T\}\cup\{A\vee p_A\Rightarrow \top\mid A\in T\}$.

\begin{lem}\label{lemma:demogan}
For any formula $A\in c(T)$, $\Psi[T] \vdash_{\mathsf{BDFNL^*}} A\wedge A^\sim\Rightarrow_{c(T)} \bot$ and $\Psi[T]$ $ \vdash_{\mathsf{BDFNL^*}} A\vee A^\sim\Rightarrow_{c(T)} \top$.
\end{lem}
\begin{proof}
We proceed by induction on the complexity of formula $A$. Assume $A\in T$. Then the claim obviously holds. Assume $A=B\wedge C$. Then $A^\sim=(B\wedge C)^\sim=B^\sim \vee C^\sim$. By inductive hypothesis, $\vdash_{\mathsf{BDFNL^*}}B\wedge B^\sim \Rightarrow_{c(T)} \bot$ and $\vdash_{\mathsf{BDFNL^*}} C\wedge C^\sim \Rightarrow_{c(T)} \bot$, whence by ($\bot$) and ($\mathrm{Cut}$), one gets $B\wedge B^\sim\Rightarrow_{c(T)} B\wedge C^\sim$ and $C\wedge C^\sim \Rightarrow_{c(T)} C\wedge B^\sim$. Hence by applying ($\vee L$) to the former one and $B\wedge C^\sim \Rightarrow_{c(T)} B\wedge C^\sim$, one obtains $(B\wedge B^\sim)\vee (B \wedge C^\sim) \Rightarrow_{c(T)} B\wedge C^\sim$. Consequently, by ($\mathrm{D}$) and ($\mathrm{Cut}$), one gets $B\wedge (B^\sim \vee C^\sim)\Rightarrow_{c(T)} B\wedge C^\sim$. By similar arguments, one gets $C\wedge (B^\sim \vee C^\sim)\Rightarrow_{c(T)} C\wedge B^\sim$. Hence by ($\wedge \mathrm{L}$), ($\wedge \mathrm{R}$) and ($\mathrm{Cut}$), one gets $ (B\wedge C)\wedge (B^\sim \vee C^\sim) \Rightarrow_{c(T)} B\wedge C^\sim \wedge C\wedge B^\sim$. By inductive hypothesis, ($\wedge \mathrm{L}$), ($\wedge \mathrm{R}$) and ($\mathrm{Cut}$), one obtains $ B\wedge C^\sim \wedge C\wedge B^\sim  \Rightarrow_{c(T)} \bot$. Hence $\vdash_{\mathsf{BDFNL^*}} (B\wedge C)\wedge (B^\sim \vee C^\sim) \Rightarrow_{c(T)}\bot$.

Assume $A= (B\vee C)$. Then $(B\vee C)^\sim= B^\sim \wedge C^\sim$. By inductive hypothesis, one gets $\vdash_{\mathsf{BDFNL^*}} B\wedge B^\sim \Rightarrow_{c(T)} \bot$ and $\vdash_{\mathsf{BDFNL^*}} C\wedge C^\sim \Rightarrow_{c(T)} \bot$. Then by ($\mathrm{\wedge L}$), one gets $B\wedge B^\sim \wedge C^\sim \Rightarrow_{c(T)} \bot$ and $C\wedge C^\sim\wedge B^\sim \Rightarrow_{c(T)} \bot$. Then by ($\mathrm{\vee L}$), one obtains $(B\wedge B^\sim \wedge C^\sim)\vee (C\wedge C^\sim\wedge B^\sim)\Rightarrow_{c(T)}\bot$.  Consequently by ($\mathrm{D}$) and ($\mathrm{Cut}$), $\Psi[T] \vdash_{\mathsf{BDFNL}^*} (B\vee C) \wedge (B^\sim \wedge C^\sim)\Rightarrow_{c(T)} \bot$. By similar arguments, one gets $\Psi[T] \vdash_{\mathsf{BDFNL}^*} A\vee A^\sim\Rightarrow_{c(T)} \top$. \qed
\end{proof}

Let $T$ be a set of $\mathcal{L}_{\mathsf{BFNL^*}}$-formulae. By $exn(T)$ we denote the subset of $T$ restricted to $\mathcal{L}_{\mathsf{BDFNL^*}}$-formulae. Then the map $(.)^\sim:(c(exn(T))\hookrightarrow c(exn(T)^\sim))$ is defined as above. Now we define an embedding function $(.)^\ddagger$ from $\mathcal{L}_{\mathsf{BFNL^*}}$-formulae to $\mathcal{L}_{\mathsf{BDFNL^*}}$-formulae inductively as follows:
\begin{itemize}
\item[(1)] $p^\ddagger=p$;
\item[(2)] $(A\star B)^\ddagger=A^\ddagger\star B^\ddagger$ for $\star\in\{\cdot,\backslash,/,\wedge,\vee\}$.
\item[(3)] $ (\neg A)^\ddagger = (A^\ddagger)^\sim$.
\end{itemize}
Intuitively, we interpret the boolean negation $\neg A$ as the formula $A^\sim$ which is a propositional letter $p_A$ for $A\in T$.
For any set $T$ of  $\mathcal{L}_{\mathsf{BFNL^*}}$-formulae, let $T^\ddagger = \{A^\ddag\mid A\in T\}$.

Let $T$ be a set of formulae closed under subformulae. By $c'(T)$ we mean the closure of $T$ under $\vee,\wedge$ and $\neg$. Obviously $(c'(T))^\ddagger= c(exn(T)^\sim)$.

Let $T$ be the set of all subformulae of formulae appearing in $\Gamma \Rightarrow A$ and contains $\top$ and $\bot$. Define $\Psi[exn(T)]= \{A\wedge p_A\Rightarrow \bot \mid A\in exn(T)\}\cup\{A\vee p_A\Rightarrow \top\mid A\in exn(T)\}$.

\begin{lem}\label{Lemma:embbdfnl}
For any $\mathcal{L}_{\mathsf{BFNL^*}}$ sequent $\Gamma \Rightarrow A$, $\vdash_{\mathsf{BFNL^*}} \Gamma \Rightarrow_{c'(T)} A$ iff $\Psi[exn(T)] \vdash_{\mathsf{BDFNL^*}} \Gamma^\ddagger$ $\Rightarrow_{c(exn(T)^\sim)} A^\ddagger$.
\end{lem}
\begin{proof}
We proceed by induction on the length of the $c'(T)$-deduction of $\Gamma \Rightarrow A$ in $\mathsf{BFNL^*}$. By the definition of $\ddagger$, ($\mathrm{Id}$), ($\bot$), ($\top$) and ($\mathrm{D}$) are obvious. ($\neg 1$) and ($\neg 2$) follows from Lemma \ref{lemma:demogan}. Since all rules in $\mathsf{BFNL}^*$ happen to be rules of $\mathsf{BDFNL}^*$, by inductive hypothesis the claim holds. For the converse direction, since all rules and axioms of $\mathsf{BDFNL}^*$ are rules and axioms in $\mathsf{BFNL}^*$ and all assumptions in $\Psi[exn(T)]$ are of the form $A\wedge p_A\Rightarrow \bot$ or $A\vee p_A\Rightarrow \top$, by the definition of $\ddagger$, a deduction of $\Gamma^\ddagger \Rightarrow A^\ddagger$ can be easily rewritten as a deduction of $\Gamma \Rightarrow A$ in $\mathsf{BFNL}^*$ by replacing all occurrences of $p_A$ by $\neg A$ for any formula $A$.
\end{proof}

The following lemma on subformula property is proved in \cite{bus4}.
\begin{lem}[\cite{bus4}]\label{lemma:subfor}
If $ \vdash_{\mathsf{BFNL^*}} \Gamma \Rightarrow A$, then $ \vdash_{\mathsf{BFNL^*}} \Gamma \Rightarrow_{c'(T)} A$
\end{lem}

By Lemma \ref{lemma:subfor} and \ref{Lemma:embbdfnl}, one obtains the following theorem immediately.
\begin{thm}\label{theorem:embbdfnl}
 $ \vdash_{\mathsf{BFNL^*}} \Gamma \Rightarrow A$ iff $\Psi[exn(T)] \vdash_{\mathsf{BDFNL^*}} \Gamma^\ddagger \Rightarrow A^\ddagger$
\end{thm}

Obviously the construction of $\Psi[exn(T)]$ and the reduction are both in polynomial time, together with Theorem \ref{theorem:embbdfnl} and \ref{theorem:comBFNL}, one gets the following theorem.
\begin{thm}
The decision problem in $\mathsf{BDFNL}^*(\mathrm{\Phi})$ is PSPACE-hard.
\end{thm}

Now let us embed $\mathsf{BDFNL^*}$ into $\mathsf{DFNL^*}$. First we define a set of special simple sequents which will be used to replace the role of $\top$ and $\bot$ in $\mathsf{BDFNL^*}$.
Let $p_\bot$ and $p_\top$ be two distinguished propositional letters. Let $T$ be a set of $\mc{L}_\mathsf{DFNL^*}$-formulae containing $p_\bot$ and $p_\top$ and closed under subformulae. By $\Theta[T]$ we mean a set of sequents containing all sequents of the following form:
\begin{displaymath}
p_\bot\Rightarrow A\quad A\circ p_\bot \Rightarrow p_\bot\quad p_\bot\circ A\Rightarrow p_\bot
\end{displaymath}
\begin{displaymath}
A\Rightarrow p_\top\quad A\circ p_\top \Rightarrow p_\top\quad p_\top\circ A\Rightarrow p_\top
\end{displaymath}
where $A\in T$. Then we may prove the following lemma.

\begin{lem}
Let $T$ be a set of $\mc{L}_\mathsf{DFNL^*}$-formulae containing $p_\bot$ and $p_\top$ and closed under subformulae. Then for all $A\in c(T)$, the sequents $p_\bot\Rightarrow A, A\circ p_\bot \Rightarrow p_\bot, p_\bot\circ A\Rightarrow p_\bot, A\Rightarrow p_\top,  A\circ p_\top \Rightarrow p_\top,  p_\top\circ A\Rightarrow p_\top$ are derivable from $\Theta[T]$ in $\mathsf{DFNL^*}$.
\end{lem}
\begin{proof}
By induction on the complexity of $A$. The case of $A\in T$ is obvious. Here we only show the proof of sequents of the first two form, others can be proved by similar arguments. Consider the sequent of the form $p_\bot \Rightarrow A$. Assume $A=A_1\wedge A_2$. By inductive hypothesis, one obtains $p_\bot \Rightarrow A_1$ and $p_\bot \Rightarrow A_2$. By ($\mathsf{\wedge R}$), one gets $p_\bot \Rightarrow A_1\wedge A_2$. Assume $A=A_1\vee A_2$. By inductive hypothesis, one obtains $p_\bot \Rightarrow A_1$, whence by ($\mathrm{\vee R}$), one gets $p_\bot\Rightarrow A_1\vee A_2$. Then let us consider the sequent of the form $A\circ p_\bot\Rightarrow p_\bot$. Assume that $A=A_1\wedge A_2$. By inductive hypothesis, one gets $A_1\circ p_\bot \Rightarrow p_\bot$. Hence by ($\mathrm{\wedge L}$), one obtains $A_1\wedge A_2\circ p_\bot \Rightarrow p_\bot$. By similar arguments, if $A=A_1\vee A_2$, then one obtains $A_1\vee A_2\circ p_\bot \Rightarrow p_\bot$. \qed
\end{proof}

\begin{lem}\label{theorem:misbottop}
Let $T$ be a set of $\mc{L}_\mathsf{DFNL^*}$-formulae containing $p_\bot$ and $p_\top$ and closed under subformulae. Then the $c(T)$-sequents $\Gamma[\bot]\Rightarrow A$ and $\Gamma \Rightarrow \top$ are derivable from $\Theta[T]$ in $\mathsf{DFNL^*}$.
\end{lem}
\begin{proof}
We prove the first sequent by induction on the total number $n$ of $\circ$ in the sequent. The second one can be show similarly. The basic case $n\leq1$ is easy. Assume $\Gamma[\bot]=\Gamma'[\Delta\circ \bot]$. By inductive hypothesis, one obtains $ \Delta\circ \bot \Rightarrow \bot$ and $ \Gamma'[\bot] \Rightarrow A$ are both derivable from $\Theta$ in $\mathsf{DFNL^*}$. Hence by (Cut), one gets $\Gamma[\bot]\Rightarrow A$. \qed
\end{proof}

We define an embedding function $(.)^\S$ from $\mathcal{L}_{\mathsf{BDFNL^*}}$-formulae to $\mathcal{L}_{\mathsf{DFNL^*}}$-formulae inductively as follows:
\begin{itemize}
\item[(1)] $\bot^\S=p_\bot$  and $\top^\S=p_\top$.
\item[(2)] $(A\star B)^\S=A^\S\star B^\S$ for $\star\in\{\cdot,\backslash,/,\wedge,\vee\}$.
\end{itemize}

Let $\Gamma \Rightarrow A$ be a $\mc{L}_{\mathsf{BDFNL^*}}$-sequent and $\Phi$ a finite set of $\mc{L}_{\mathsf{BDFNL^*}}$-sequents. Let $T$ be the set of all subformulae occured in $\Gamma\Imp A$ or $\Phi$, and containing $\top$ and $\bot$. First we recall the following lemma from \cite{bus4}.
\begin{lem}
If $\Phi \vdash_{\mathsf{BDFNL^*}} \Gamma \Rightarrow A$, then $\Phi \vdash_{\mathsf{BDFNL^*}} \Gamma \Rightarrow_{c(T)} A$.
\end{lem}

By $ec(T)$ we mean the set obtained from $T$ by replacing all occurrences of $\bot, \top$ in formulae by $p_\top, p_\bot$. Notice that $(c(T))^\S=c(ec(T))$. Let $\Theta[ec(T)]$ be the set of all sequents $p_\bot\Rightarrow A, A\circ p_\bot \Rightarrow p_\bot, p_\bot\circ A\Rightarrow p_\bot, A\Rightarrow p_\top,  A\circ p_\top \Rightarrow p_\top,  p_\top\circ A\Rightarrow p_\top$ for $A\in ex(T)$.
Since all rules of $\mathsf{BDFNL^*}$ are rules of $\mathsf{DFNL^*}$, together with  Lemma \ref{theorem:misbottop} one can easily obtain the following lemma.
\begin{lem}
$\Phi \vdash_{\mathsf{BDFNL^*}} \Gamma \Rightarrow_{c(T)}  A$ iff $\Phi\cup\Theta[ec(T)] \vdash_{\mathsf{DFNL^*}} \Gamma^\S \Rightarrow_{c(ec(T))} A^\S$
\end{lem}

Now we conclude with the following theorem

\begin{thm}\label{theorem:embdfnl}
$\Phi \vdash_{\mathsf{BDFNL^*}} \Gamma \Rightarrow A$ iff $\Phi\cup\Theta[ec(T)] \vdash_{\mathsf{DFNL^*}} \Gamma^\S \Rightarrow A^\S$.
\end{thm}

Obviously both the construction of the set $\Phi\cup\Theta[ec(T)]$ and the reduction are in polynomial time. Then by Theorem \ref{theorem:embdfnl} and \ref{theorem:comBFNL}, one gets the following theorem.

\begin{thm}
The decision problem of $\mathsf{DFNL}^*(\mathrm{\Phi})$ is PSPACE-hard.
\end{thm}

\section{Some Variants of $\mathsf{BFNL^*}$}
Let us apply the methods in section one to some variants of $\mathsf{BFNL^*}$. The first example is  $\mathsf{BFNL_e^*}$, i.e. $\mathsf{BFNL^*}$ with the following exchange rule:
\begin{displaymath}
(\cdot E)\quad \frac{\Gamma[\Delta_1\circ \Delta_2]\Rightarrow A}{\Gamma[\Delta_2\circ \Delta_1]\Rightarrow A}
\end{displaymath}
In $\mathsf{BFNL_e^*}$, $A\backslash B \Leftrightarrow A/B$ holds and hence we consider only one residual usually denoted $A\rightarrow B$. All results from section 1 can be proved for $\mathsf{BFNL_e^*}$.  The embedding function $\dag$ and the proofs of Lemma \ref{lemma:KBFNL} remains the same. However the construction for the ternary relation model ($\mf{J}^\mf{M}$) for $\mathsf{BFNL_e^*}$ requires some modifications in order to satisfy that $\mf{J}^\mf{M}\models A\cdot B$ iff $\mf{J}^\mf{M} \models B\cdot A$.

Let $\mf{M}=(W, R, V)$ be a Kripke model for $\mathsf{K}$. Define an $\mathsf{BFNL_e^*}$-model $\mathfrak{J^M}=(W', R', V')$ from $\mf{M}$ as follows:
\begin{itemize}
\item[(1)] $W'=\{u_1,u_2| u\in W\}$
\item[(2)] $R'=\{\langle v_1,u_1,u_2\rangle, \langle v_1,u_2, u_1\rangle, \langle v_2, u_1, u_2\rangle, \langle v_2,u_2,u_1\rangle\mid vRu\}$
\item[(3)] $V'(p) = \{u_1,u_2\mid u\in V(p)\}$ for $p\in\mathsf{Prop}$; and $V'(m) = W'$.
\end{itemize}

Lemma \ref{lemma:truth} remains ture. In order to get an analogous theorem of Theorem \ref{thm:comBFNL^*}, we need the following two lemmas.

\begin{lem}\label{ARBAtruth}
For any $\mathcal{L}_\mathsf{BFNL_e^*}$-formula $A$ and $u_1,u_2\in W'$, $\mf{J}^\mf{M},u_1\models A$ iff $\mf{J}^\mf{M},u_2\models A$.
\end{lem}
\begin{proof}
We proceed by induction on the complexity of $A$. The cases of atomic formulae, $A\wedge B$ and $A\rightarrow B$ are easy. We show only the cases of $A\cdot B$ and $A\rightarrow B$. Assume
$\mf{J}^\mf{M}, v_1\models A\cdot B$. By construction, there exist $u_1,u_2\in W'$ such that $R'(v_1, u_1,u_2)$ and $\mf{J}^\mf{M},u_1\models A$ and $\mf{J}^\mf{M},u_2\models B$. By the construction, $R'(v_2, u_1,u_2)$. Consequently, $\mf{J}^\mf{M}, v_2\models A\cdot B$. The other direction is shown similarly.
Assume $\mf{J}^\mf{M}, u_1\models A\rightarrow B$. By construction, for any $v_i \in W'$, one obtains $R'(v_i,u_2,u_1)$ and $\mf{J}^\mf{M},u_2\models A$ and $\mf{J}^\mf{M},v_i\models B$. Suppose $v_i= v_1$ without loss of generality. By inductive hypothesis, $\mf{J}^\mf{M},u_1\models A$. Since by construction $R'(v_1,u_1,u_2)$, one gets $\mf{J}^\mf{M},u_2\models A\rightarrow B$. The other direction is shown similarly.\qed
\end{proof}

\begin{lem}\label{lemma:ex}
$\mf{J}^\mf{M}\models A\cdot B\Leftrightarrow B\cdot A$.
\end{lem}
\begin{proof}
We prove the left to right direction. The other direction can be shown similarly.
Assume that $\mf{J}^\mf{M}, v_1\models A\cdot B$. Then there exist $u_1,u_2\in W'$ such that $R'(v_1,u_1,u_2)$, $\mf{J}^\mf{M}, u_1\models A$ and $\mf{J}^\mf{M}, u_2\models B$. By Lemma \ref{ARBAtruth}, one obtains $\mf{J}^\mf{M}, u_1\models B$ and $\mf{J}^\mf{M}, u_2\models A$. Hence $\mf{J}^\mf{M}, v_1\models A\cdot B$. \qed

\end{proof}

All results of section 1 can also easily be adapted for the modal extensions of $\mathsf{BFNL^*_i}$ ($i\in\{\mathsf{K,T,K4,S4,S5}\}$).  Now formula trees that occur in the antecedents of sequents are composed from formulae by two structure operations, a binary one $\circ$ and a unary one $\langle-\rangle$, corresponding to the two products $\cdot$ and $\Diamond$, respectively. Caution the language of modal formulae contains $\Diamond A$, $\Box^\downarrow A$ and formula trees contains $\langle \Gamma \rangle$. $\mathsf{BFNL^*_i}$ is obtained from $\mathsf{BFNL^*}$ by adding the following modal rules and $i$ modal logic axioms, respectively.
\begin{displaymath}
(\mathrm{\Diamond L}) \quad\frac{\Gamma[\langle A \rangle] \Rightarrow B}{\Gamma[\Diamond A] \Rightarrow B} \quad\quad (\mathrm{\Diamond R})\quad \frac{\Gamma \Rightarrow A}{\langle \Gamma \rangle \Rightarrow \Diamond A}
\end{displaymath}
\begin{displaymath}
(\mathrm{\Box^\downarrow L}) \quad\frac{\Gamma[A] \Rightarrow B}{\Gamma[\langle \Box^\downarrow A \rangle] \Rightarrow B} \quad(\mathrm{\Box^\downarrow R}) \quad\frac{\langle \Gamma \rangle \Rightarrow A}{\Gamma \Rightarrow \Box^\downarrow A}
\end{displaymath}
\begin{displaymath}
\mathrm{(T)} \quad A\Rightarrow \Diamond A\quad \mathrm{(4)}\quad \Diamond \Diamond A \Rightarrow \Diamond A \quad \mathrm{(5)}\quad \Diamond A \Rightarrow \Box \Diamond A
   \end{displaymath}
By the results in \cite{Me22}, we know that all $\mathsf{BFNL^*_i}$ admit subformula property. Noticed that axiom (K) $\Box (A \supset B)\Rightarrow \Box A \supset \Box B$, where $\Box=\neg \Diamond \neg$, is admissible in $\mathsf{BFNL^*}$ enriched with the above $\Diamond$ and $\Box^\downarrow$ rules. It is sufficed to show these modal extensions of BFNL$^*$ are conservative extensions of $\mathsf{BFNL^*}$. Then the proofs of PSPACE-hardness of the the decision problems in these systems follow from Theorem \ref{thm:comBFNL^*}.

\begin{lem}
For any $\mathcal{L}_{\mathsf{BFNL^*}}$ sequent $\Gamma \Rightarrow A$, $\vdash_{\mathsf{BFNL^*}}\Gamma \Rightarrow A$ iff $\vdash_{\mathsf{BFNL^*_i}}\Gamma \Rightarrow A$.
\end{lem}
\begin{proof}
The 'if' part is easy. We show the 'only if' part. Assume that $\vdash_{\mathsf{BFNL^*_i}}\Gamma \Rightarrow A$. By subformula property, there exists a derivation containing no modal formulae, which yields that no $\Diamond$-rules and $\Box^\downarrow$-rules are applied in this derivation. It also follows that no modal axioms appear in this derivation. Hence this derivation can be treated as a derivation in $\mathsf{BFNL^*}$. Hence $\vdash_{\mathsf{BFNL^*}}\Gamma \Rightarrow A$. \qed
\end{proof}

Since the reduction is trivial, one gets the following theorem.
\begin{thm}
The decision problems in $\mathsf{BFNL^*_i}$ for $i\in \{\mathsf{K, T, K4, S4,S5}\}$ are PSPACE-hard.
\end{thm}

This result can also be proved for $\mathsf{BFNL^*_{ei}}$, and proofs are similar as above. One can add the multiplicative constant 1. We consider the axiom $(1\mathrm{R})$ $\Rightarrow 1$, and  the rules:
\begin{displaymath}
\mathrm{(1L}_l)\quad \frac{\Gamma[\Delta]\Rightarrow A}{\Gamma[1\circ \Delta]\Rightarrow A}\quad \mathrm{(1L}_r)\quad  \frac{\Gamma[\Delta]\Rightarrow A}{\Gamma[ \Delta\circ 1]\Rightarrow A}.
\end{displaymath}
There are no problems with adapting our results for $\mathsf{BFNL1_i}$ and $\mathsf{BFNL1_{ei}}$, i.e $\mathsf{BFNL_i^*}$ with 1 and $\mathsf{BFNL_{ei}^*}$ with 1. The only difference is that the construction of $\mf{J^M}$ required additional conditions. One adds a specail element $1$ to $W'$ such that for any $u\in W'$, $R'(u,1,u)$ and $R'(u,u,1)$ hold. Moreover, for any propositional letter $p$, $1\in V'(p)$ iff $V(p)=W$. By induction on the complexity of formulae, one can easily prove that $\mf{J^M}\models A$ iff  $\mf{J^M}, 1\models A$. On the other hand, these new conditions do no effect on the Lemma \ref{lemma:truth} and Lemma \ref{lemma:com}. Hence our proof of PSPACE-hardness remains true, which yields the decision problems for $\mathsf{BFNL1_i}$, $\mathsf{BFNL1_{ei}}$, $\mathsf{BFNL1_i}$ and $\mathsf{BFNL1_{ei}}$ are PSPACE-hard.
\bibliographystyle{plain}
\bibliography{Linzhelacl}

\end{document}